\documentclass{article}

\usepackage[margin=1.0in]{geometry}

\usepackage[noblocks]{authblk}

\usepackage[english]{babel}

\usepackage{braket}
\usepackage{cancel,amsmath, amsthm, amsfonts, amssymb, graphicx, booktabs, braket, bbold}
\usepackage[all]{xy}

\newcommand{\tr}{\mbox{\textnormal{Tr}}}
\newcommand{\unif}[1]{\frac{\mbox{\textnormal{id}}_{#1}}{d_{#1}}}

\newtheorem{definition}{Definition}
\newtheorem{lemma}{Lemma}
\newtheorem{proposition}{Proposition}
\newtheorem{theorem}{Theorem}
\newtheorem{remark}{Remark}
\newtheorem{corollary}{Corollary}

\newcommand{\squareforqed}{$\square$}

\date{}

\begin{document}
\title{Recoverability from direct quantum correlations}
\author[1,2]{S. Di Giorgio}
\author[1,2]{P. Mateus}
\author[1,2]{B. Mera}

\affil[1]{Departamento de Matemática, Instituto Superior Técnico, Universidade de Lisboa, Av. Rovisco Pais, 1049-001 Lisboa, Portugal}
\affil[2]{Instituto de Telecomunica\c{c}\~oes, 1049-001 Lisboa, Portugal}

\maketitle
\begin{abstract}
We address the problem of compressing density operators defined on a finite dimensional 	Hilbert space which assumes a tensor product decomposition. In particular, we look for an \emph{efficient procedure} for learning the \emph{most likely} density operator, according to Jaynes' principle, given a \emph{chosen set} of partial information obtained from the unknown quantum system we wish to describe. For complexity reasons, we restrict our analysis to \emph{tree-structured} sets of bipartite marginals. We focus on the tripartite scenario, where we solve the problem for the couples of measured marginals which are compatible with a quantum Markov chain, providing then an algebraic necessary and sufficient condition for the compatibility to be verified. 
We introduce the generalization of the procedure to the n-partite scenario, giving some preliminary results. In particular, we prove that if the pairwise Markov condition holds between the subparts then the choice of the \emph{best} set of tree-structured bipartite marginals can be performed efficiently. Moreover, we provide a new characterisation of quantum Markov chains in terms of \emph{quantum Bayesian updating processes}. 
\end{abstract}
\section{Introduction}
\label{intro}

The problem of efficiently compressing density operators can be related to the one addressed by Jaynes~\cite{jay:57} for probability distributions. We are interested in determining an \emph{efficient procedure} for inferring the \emph{most likely density operator} from partial information about the system we wish to describe, with the freedom of choosing the \emph{set of partial information} to be collected.
With a complete set of measurements, quantum tomography techniques are able to infer with maximum accuracy the density operator that most likely describes the given quantum system, but the needed resources to perform the former increases exponentially with the number of degrees of freedom of the system. A clever choice of a partial set of measurements should optimize the data collection, lead up to an efficient learning procedure and keep a good accuracy. The necessity of dealing with a statistically relevant number of degrees of freedom motivates machine learning~\cite{bis:06} and quantum machine learning techniques~\cite{bia:wit:pan:reb:wie:llo:17,sch:sin:pet:15,car:tro:17}.\\

\paragraph{The maximum entropy estimator} In his seminal paper~\cite{jay:57}, Jaynes wrote:

``Information theory provides a constructive criterion for setting up probability distribution on the basis of partial knowledge and leads to a type of statistical inference which is called the maximum-entropy estimate. It is the least biased estimate possible on the given information, i.e. it is maximally noncommittal with regard to missing information.''

Since quantum information theory provides a well-defined generalization the Shannon entropy, the von Neumann entropy, it can be used to state a quantum Jaynes' principle and, therefore, to obtain a maximally noncommittal estimator for density operators with regard to the partial information collected. Moreover, due to the concavity of both Shannon and von Neumann entropy, the maximization problem has a unique solution and the desired estimator is uniquely determined. We choose therefore to infer from the given measurements, the density operators that maximizes the von Neumann entropy. 
\paragraph{Learning from direct correlations} A common approximation for multipartite physical system description consists in cutting the correlations after the first neighbour. Direct dependencies between random variables are usually less struggling to be measured and it results in a theoretical exponential gain in data collection. Consider a multipartite quantum system, described by a Hilbert space $\mathcal{H}_{X_1,...,X_n}=\mathcal{H}_{X_1}\otimes \dots \otimes \mathcal{H}_{X_n}$, with $\dim \mathcal{H}_{X_i}=\mbox{O}(d)$, for all $i=1,...,n$. To infer the state of the whole system, represented by a density matrix $\rho_{X_1,...,X_n}$, one needs an exponential amount of resources, more concretely, $\mbox{O}(d^{2n})$. However, if one restricts to bipartite correlations only, the amount of resources needed to approximately reproduce the state scales polynomially with $n$, $\mbox{O}(n^2 d^4)$. This approximation restricts the set of learnable states: for pure-multipartite correlated systems, such as the GHZ state, we expect not to be accurately recoverable.
\paragraph{Restriction to trees} Since density operators generalize classical probability distributions, finding the density operator that maximizes the von Neumann entropy given the complete set of bipartite correlations, would solve the analogous problem for classical probability distributions. To study an efficient learning procedure, the hardness results of the classical problem need to be taken in account.

Given a finite set of classical random variables, the optimization problem of finding the probability distribution that maximizes the Shannon entropy given a general collection of bipartite marginals, can be stated as a graph inference problem. With the provided information, the classical system can be represented by a graph where the vertices represent the random variables and the edges represent the direct dependencies between them. Then, graphical models, such as Bayesian Networks and Markov random fields~\cite{pea:11,bis:06}, provide the learning techniques for the maximum entropy estimator. Inferring a general graphical structure is NP-Hard~\cite{chi:96}, and so is finding an approximate solution~\cite{dag:lub:93}. The only structures for which a general efficient learning solution is known are trees, as even learning 2-polytrees is NP-Hard~\cite{das:99}. Nevertheless, there exists an efficient algorithm to obtain the \emph{optimal tree} -- the Chow-Liu algorithm~\cite{cho:liu:68}. Then, any set of random variables can be efficiently approximated by the probability distribution describing its most likely tree. It has been an open problem to find richer structures than tree Markov random fields that can be learned efficiently.

Motivated by these classical results, we look for an efficient learning procedure for the maximum von Neumann entropy estimator given a subset of bipartite marginals which is tree-structured. This means that, representing the joint quantum system by a graph where the vertices label the subsystems and where the edges represent the measured bipartite marginals, the resulting graph is a tree. Observe that the problem we are analyzing can be stated also as searching for a subclass of density operators for which the classical results of learning via graphical models can be extended.

The efficiency of the afforementioned learning techniques is mainly due to the factorization of the joint probability distribution that occurs when the conditional independence condition between the interested subparts holds. As we are going to see in detail, also when the system involves quantum correlations, the quantum generalization of conditional independence to quantum states results in an algebraic recovery of the joint in terms of the interested subparts.
However, whereas the graphical structure of a classical system naturally encodes the conditional independence properties between them, the one involving quantum correlations does not. Further conditions, not explicit from the graphical structure, need in general to be verified. Many attempts have been done for developing appropriate generalizations of graphical models for density operators~\cite{leif:spek:13,fit:jon:ved:vla:15}, but none of them naturally encodes the required properties which result in a learning simplification without further conditions~\cite{hor:dom:al.:17}.

\subsection*{Problem} We wish to find an \emph{efficient procedure} for learning the density operator that maximizes the von Neumann entropy from a subset of (compatible, cf. Definition~\ref{def:comp}) tree-structured bipartite marginals. 
\paragraph{Results}  We consider the simplest nontrivial tree, i.e. a tripartite quantum system where two marginals are known (Section~\ref{subsection: the tripartite case}). This analysis provides insight for the multipartite scenario (Section~\ref{subsection: the qcl algorithm}).

\emph{Tripartite case}: We provide an algebraic recovery procedure of the tripartite density operator given two bipartite marginals when they are compatible with a quantum Markov chain -- Definition~\ref{def:qmc}. We find that there exists an algebraic recovery procedure, namely the Petz recovery map, for the maximum entropy estimator whenever there exists a quantum Markov chain in the compatibility set of the provided marginals -- Theorem~\ref{th:1}. Indeed, Quantum Markov chains represent the subset of tripartite density operators that strictly contains classical probability distributions, i.e., such that between the two non adjacent quantum states, the (quantum) conditional independence condition holds -- Corollary~\ref{cor:1}. Then, we provide a necessary and sufficient condition for efficiently verifying if the given couple of marginals is compatible with a quantum Markov chain --Theorem~\ref{th:2}. 
Moreover, we give a criterion for determining which couple of possible bipartite marginals, out of the three possible, provides \emph{the best} maximum entropy estimator, meaning the one that minimizes the relative entropy distance with the \emph{unknown} density operator. We prove the best estimator to be the one with minimum von Neumann entropy -- Theorem~\ref{th:5}, which, if all the three couple of measured marginals are compatible with a quantum Markov chain, can be obtained by discarding the marginal with minimum quantum mutual information -- Theorem~\ref{th:6}.
Both the conditions of Theorem~\ref{th:2} and Theorem~\ref{th:6} are algebraic, allowing the efficiency of the entire learning procedure and, possibly, an easier generalization to the multipartite case.\\
Furthermore, observing that the problem of entropy maximization can be stated as multiple-step minimum entropy updating (Section~\ref{section: quantum bayesian updating}), we give a further characterization of quantum Markov chains as the commutativity of a diagram of quantum Bayesian updating processes -- Theorem~\ref{th:4}.

\emph{Multipartite case}: We provide some preliminary results about a possible generalization of the notion of the procedure to the multipartite scenario. In particular, we prove that the global Markov condition, properly extended to density operators, is sufficient to have an efficient recovery of the maximum entropy estimator given a tree structured set of bipartite marginals and, additionally, it is sufficient to efficiently choose an optimal tree for the estimator -- Proposition~\ref{prop:3}.

\section{Maximum Entropy Estimator}
\label{section: maximum entropy estimator}
Let $\mathcal{X}=\{X_{1},...,X_{n}\}$, with $0<n<\infty$, be a set labelling the parts of a multipartite quantum system $\mathcal{X}$. The physical system $\mathcal{X}$ can be described by an Hermitian, positive-semidefinite and trace one operator $\rho_\mathcal{X}$, namely a density operator in the Liouville space $\mathcal{L}(\mathcal{H}_\mathcal{X})$, where $\mathcal{H}_{\mathcal{X}}$ is a separable Hilbert space on its subparts $\mathcal{H}_{\mathcal{X}}:=\bigotimes_{i=1}^{n}\mathcal{H}_{X_i}$. We denote  by $d_X$ the dimension of the Hilbert space $\mathcal{H}_X$, which is always assumed to be finite in the whole manuscript.
Running several times the same experiment, we can collect many copies of the unknown system; on each of them, we can perform a measurement, obtaining the set of expectation values $\{\langle\Theta_{i}\rangle\}_{i\in I}$, where $\Theta_i$ are positive Hermitian operators acting on the full joint Hilbert space $\mathcal{H}_{\mathcal{X}}$.

\begin{proposition}
\label{prop:1}
 The density operator $\widetilde{\rho}\in\mathcal{L}(\mathcal{H}_{\mathcal{X}})$ that maximizes the von Neumann entropy, denoted $S(\rho)$, is given by 
\begin{align}
\widetilde{\rho}_\mathcal{X}=\frac{1}{Z} \exp\Big(\sum_{i\in I} \lambda_i \Theta_i \Big), \textnormal{ with } Z=\tr\Big[\exp\Big(\sum_{i\in I} \lambda_i \Theta_i\Big)\Big]
\end{align}
and in which $\{\lambda_i\}_{i\in I}$ are Lagrange multipliers which are obtained by solving the equations $\tr \Big[\rho \,  \Theta_i\Big]=\langle \Theta_i\rangle$.
\end{proposition}

\begin{proof} This follows from taking the variation of the function
\begin{align}
S(\rho) -\sum_{i\in I}\lambda_i\big(\tr \Big[\rho \, \Theta_i\Big]-\langle \Theta_i\rangle \big).
\end{align}
By concavity of $S$ the solution is a maximum point and it is unique.
\end{proof}
In our case, we are given a set of bipartite marginals $\{\rho_{X_{i}X_{j}}\}$, which can be probed by a complete set of observables in the associated bipartite Hilbert spaces. If we are given a Hilbert space $\mathcal{H}_{X}$, then there exist a set of Hermitian operators $\{\Lambda_j^{(X)}: j=0,...,d_X^2-1\}$ which are complete in the sense that any linear operator and, in particular any observable, can be written as a linear combination of the latter. This basis of operators can be chosen to be orthonormal with respect to the Hilbert-Schmidt inner product and, additionally, to be traceless, so that $\{i\Lambda_{j}^{(X)}:j=0,...,d_X^2-1\}$ span a Lie algebra $\mathfrak{su}(d_X)$, where $\Lambda_{0}^{(X)}:=\mbox{id}_{X}$ is the identity on $\mathcal{H}_X$ and it corresponds to the remaining generator of $\mathfrak{u}(d_X)$. For each $i\in I$, we denote by $\{\Lambda^{(X_i)}_{k}:k=0,..,d_{X_{i}}^2-1\}$ a chosen complete set of observables for $X_i$, as before. For each $X_iX_j$, the set $\{\Lambda^{(X_i)}_{k}\otimes \Lambda^{(X_j)}_{l}: k=0,...,d_{X_{i}}^2-1,\ l=0,...,d_{X_{j}}^2-1\}$ forms a complete set of observables for the bipartite system $X_iX_j$. We extend the operators $\Lambda^{(X_i)}_{k}\otimes \Lambda^{(X_j)}_{l}$ acting on $\mathcal{H}_{X_i}\otimes \mathcal{H}_{X_j}$ in a natural way to act on the joint Hilbert space $\mathcal{H}_{\mathcal{X}}$ by taking the tensor product with the identity on the relevant factors. By abuse of notation we denote by $\Lambda^{(X_i)}_{k}\Lambda^{(X_j)}_{l}$ the extended operators. As a consequence of Prop.~\ref{prop:1}, the density operator $\widetilde{\rho}_{X_1 ... X_n}$ maximizing the von Neumann entropy assumes the form
\begin{align}
\widetilde{\rho}_{\mathcal{X}}=\frac{1}{Z}\exp\left(\sum_{i,j}\underset{k^2+l^2\neq 0}{\sum_{k=0}^{d_{X_i}^2-1}\sum_{l=0}^{d_{X_j}^2-1}}\lambda^{(X_iX_j)}_{kl}\Lambda^{(X_i)}_{k}\Lambda^{(X_j)}_{l} \right),
\label{eq: merho for 2}
\end{align}
where, as before, $\{\lambda_{jk}^{(X_iX_j)}\}$ are the Lagrange multipliers of the optimization problem, constrained by the partial traces on the given marginals. The state of Eq.~\eqref{eq: merho for 2} exists because the set of density operators satisfying the given constraints is convex and non-empty. Indeed, \emph{a priori}, the set of bipartite marginals is the output of a set of measurements performed on an existing quantum system described by an $n$-partite density operator. We are than assuming the given set of provided marginals to be \emph{compatible}.
\begin{definition}
\label{def:comp}
(Compatibility and Compatibility set) 
Consider a set of density operators $\mathcal{C}=\{\rho_{\mathcal{Y}}\in\mathcal{L}(\mathcal{H}_{\mathcal{Y}}) \}_{\mathcal{Y}\in\mathcal{K}}$, where $\mathcal{K}$ is a family of subsystems of $\mathcal{X}$ which is a cover, i.e., $\bigcup_{\mathcal{Y}\in\mathcal{K}}\mathcal{Y}=\mathcal{X}$. We say that $\mathcal{C}$ is a compatible set of marginals if there exists at least one density operator $\rho$ over the joint Hilbert space $\mathcal{H}_{\mathcal{X}}$ such that $\tr_{\overline{\mathcal{Y}}}(\rho)=\rho_{\mathcal{Y}}$ for all $\rho_{\mathcal{Y}}\in \mathcal{C}$. Here $\tr_{\overline{\mathcal{Y}}}(\cdot)$ denotes the partial trace over the complementary factors of the Hilbert space $\bar{\mathcal{Y}}=\mathcal{X}\backslash\mathcal{Y}$.\\
Moreover, we denote by \textnormal{Comp}$(\mathcal{C})$ the set of density operators over $\mathcal{H}_{\mathcal{X}}$ s.t. $\tr_{\overline{\mathcal{Y}}}(\rho)=\rho_{\mathcal{Y}}$ for every $\rho_{\mathcal{Y}}\in\mathcal{C}$. We say for each $\rho\in$ \textnormal{Comp}$(\mathcal{C})$ that $\rho$ is compatible  with $\rho_{\mathcal{Y}}$, for any $\mathcal{Y}\in\mathcal{K}$.  Additionally, we also say that $\rho$ is compatible with $\mathcal{C}$.
\end{definition}
The problem of determining if a given set of density operators is compatible is known as \emph{the quantum marginal problem}  is  QMA complete~\cite{liu:06,liu:chr:ver:07}. This problem reduces to determining the maximum entropy estimator (compatible with the marginals), and therefore, the latter is QMA hard.
In particular, a necessary condition, also sufficient for classical probability distributions, for a set of density operators defined on overlapping Hilbert spaces to be compatible is to coincide on their intersection. Formally, given $\rho_{\mathcal{Y}_1},\rho_{\mathcal{Y}_2}\in \mathcal{C}$ compatible, if $\mathcal{Y}_1\cap \mathcal{Y}_2\neq \emptyset$ then  $\tr_{\overline{\mathcal{Y}_1\cap \mathcal{Y}_2}}(\rho_{\mathcal{Y}_1})=\tr_{\overline{\mathcal{Y}_1\cap \mathcal{Y}_2}}(\rho_{\mathcal{Y}_2})$. 

Our goal, as stated in the introduction, is to have an efficient recovery of the maximum entropy estimator given a set of bipartite marginals. For that reason and taking into account the results known for the classical case (see the introduction), we restrict to the case where the set of bipartite marginals is tree structured. In the next subsection we consider the tripartite case in detail.
\begin{remark}
Also in the tripartite case, we are going to restrict the problem to trees where just two marginals out of the three possible are taken in account. This not only to gain some insight about the generalization to the multipartite scenario, but also since the recovery problem for a tripartite probability distribution given all the three possible bipartite marginals is open \cite{hall:07,che:git:mod:pia:14,son:jaa:08}. Moreover, moving to the quantum scenario, also the compatibility problem for just a couple of overlapping marginals is open~\cite{tyc:vla:15,hub:2018}. We are then going to assume the set of the two given marginal density operators compatible.
\end{remark}
\subsection{The tripartite case}
\label{subsection: the tripartite case}
Let us denote $\mathcal{X}=\{A,B,C\}$, and assume we are given access to marginals $\{\rho_{AB},\rho_{BC}\}$. See the associated graph in Fig.~\ref{fig:1}.
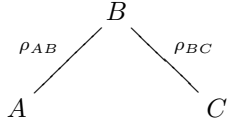
\begin{figure}[h!]
\begin{displaymath}
    \xymatrix{ & B \ar@{-}[dl]_{\rho_{AB}} \ar@{-}[dr]^{\rho_{BC}} & \\
              A & & C }
\end{displaymath}
\caption{Graph associated to system of bipartite marginals $\{\rho_{AB},\rho_{BC}\}$. Each vertex represents a quantum system to which is associated a density operator by partial tracing the marginals. The edges correspond to direct correlations between the vertices, associated to mixed states defined over the edge Hilbert space.}
\label{fig:1}
\end{figure}
In Eq.~\eqref{eq: merho for 2}, we denote by $\lambda_{kl}:=\lambda_{kl}^{(AB)}$ and $\eta_{kl}:=\lambda_{kl}^{(BC)}$. We then have,
\begin{align}
\label{eq:4}
\widetilde{\rho}_{ABC}=	\frac{1}{Z}\exp\left(\underset{k^2+l^2\neq 0}{\sum_{k=0}^{d_{A}^2-1}\sum_{l=0}^{d_{B}^2-1}}\lambda_{kl}\Lambda_{k}^{A}\Lambda_{l}^{B}+\underset{k^2+l^2\neq 0}{\sum_{k=0}^{d_{B}^2-1}\sum_{l=0}^{d_{C}^2-1}}\eta_{kl}\Lambda_{k}^{B}\Lambda_{l}^{C}\right).
\end{align}
Because for each $X_iX_j$, the set $\{\Lambda^{(X_i)}_{k}\otimes \Lambda^{(X_j)}_{l}: k=0,...,d_{i}^2-1,\ l=0,...,d_{j}^2-1\}$ forms a complete set of observables for the bipartite system $X_iX_j$, we can write,
\begin{align}
\rho_{AB}=\frac{1}{d_A d_B}\mbox{id}_{AB}+	\underset{k^2+l^2\neq 0}{\sum_{k=0}^{d_{A}^2-1}\sum_{l=0}^{d_{B}^2-1}}\alpha_{kl}\Lambda_{k}^{(A)}\Lambda_{l}^{(B)},\quad \rho_{BC}=\frac{1}{d_B d_C}\mbox{id}_{BC}+	\underset{k^2+l^2\neq 0}{\sum_{k=0}^{d_{A}^2-1}\sum_{l=0}^{d_{B}^2-1}}\beta_{kl}\Lambda_{k}^{(B)}\Lambda_{l}^{(C)},
\end{align}
where $\alpha_{kl}, \beta_{kl}\in\mathbb{R}$. Additionally, the constraints imposed by the marginals can be cast in the form
\begin{align}
\label{eq:7}
\tr\left[\Lambda_{k}^{(A)}\Lambda_{l}^{(B)}\left(\widetilde{\rho}_{ABC}-\rho_{AB}\otimes\unif{C} \right)\right]=0,\quad \tr\left[\Lambda_{k}^{(B)}\Lambda_{l}^{(C)}\left(\widetilde{\rho}_{ABC}-\unif{A}\otimes\rho_{BC} \right)\right]=0,
\end{align}
where $k$ and $l$ range in the dimensions of the space of linear operators of the appropriate subsystem.

From the partition function $Z$, one can obtain, by differentiation, a system of equations which allow to solve for the Lagrange multipliers in terms of the known parameters $\{\alpha_{kl}\}$ and $\{\beta_{kl}\}$. This is a standard procedure in statistical mechanics. Determining if the optimization problem described above has an algebraic solution is not trivial: the Lie operators in the exponent of Eq.~\eqref{eq:4} in general do not commute, which makes the analysis of the constraints imposed in Eq.~\eqref{eq:7} mathematically hard to manage. Many techniques for overcoming similar problems are object of study, see for example Ref.~\cite{sut:18} for trace inequalities or Refs.~\cite{kim:12,rus:13} for operator inequalities. Instead, here we focus our attention on a well-behaved subset of density operators -- quantum Markov chains~\cite{faw:ren:15}. In the next section, we recall the definition of a quantum Markov chain and prove a necessary and sufficient condition for a given pair of bipartite marginals to be compatible with a quantum Markov chain. 
\section{Marginals' compatibility with quantum Markov chains}
\label{section: marginals' compatibility with qmc}
\begin{definition}\label{def:qmc}
 (Quantum Markov chain~\cite{sut:18}) A tripartite state $\rho_{ABC}$ over $\mathcal{H}_{ABC}$ is called a quantum Markov chain (QMC) in order $A-B-C$ if there exists a recovery map $\mathcal{R}_{B\to BC}:\mathcal{L}(\mathcal{H}_{B})\to \mathcal{L}(\mathcal{H}_{BC})$ such that
\begin{align}
\rho_{ABC}=\left(\mathcal{I}_{A}\otimes\mathcal{R}_{B\to BC}\right)(\rho_{AB}),
\end{align} 
where a recovery map is an arbitrary trace-preserving completely positive (CPTP) map, see Ref.~\cite{bis:06}, and $\mathcal{I}_{A}$ denotes the identity map on $\mathcal{L}(\mathcal{H}_A)$.
\end{definition}

A QMC can be characterized from an information-theoretical point of view due to the following result.

\begin{proposition}
\label{prop:2}
\cite{hay:joz:pet:win:04}
A tripartite state $\rho_{ABC}$ is a QMC in the order $A-B-C$ if an only if $ I_{\rho}(A:C|B)=0$, where $I_{\rho}(A:C|B):=S(\rho_{AB})+S(\rho_{BC})-S(\rho_{B})-S(\rho_{ABC})$ is the quantum conditional mutual information.
\end{proposition}
We recall that in case of classical random variables with finite domains, when we are given two bipartite marginals $p(A,B)$ and $p(B,C)$, the compatibility condition $\sum_ap(A=a,B=b)=\sum_cp(B=b,C=c)$ is necessary and sufficient for $A$ and $C$ to be independent conditioned on $B$, i.e. $I(A: C|B)=0$. Then, the set of quantum Markov chains includes the set of classical tripartite probability distributions.

Proposition~\ref{prop:2} is equivalent to the statement: a QMC $A-B-C$ is a tripartite quantum state for which the strong subadditivity of the von Neumann entropy,
\begin{align}
S(\rho_{AB})+S(\rho_{BC})\geq S(\rho_{B})+S(\rho_{ABC}),
\end{align}
holds with equality~\cite{lie:rus:73}, which trivially results in the fact that a QMC $A-B-C$ maximizes the von Neumann entropy given its two bipartite marginals $AB$ and $BC$. Furthermore, the quantum systems $A$ and $C$ are said to be \emph{quantum conditionally independent} given the quantum system $B$.

A QMC always admits as a recovery channel the \emph{rotated Petz recovery map}~\cite{pet:86,wil:15}:
\begin{align}
\mathcal{P}_{B\to BC}^{t}(X):=\rho_{BC}^{\frac{1+it}{2}}\left(\rho_B^{-\frac{1+it}{2}}X \rho_B^{-\frac{1-it}{2}}\right)\rho_{BC}^{\frac{1-it}{2}}, \text{ for any } X\in \mathcal{L}(\mathcal{H}_{B}),\ t\in\mathbb{R}.
\end{align}
In particular, for $t=0$, the map is known as the \emph{Petz recovery map or transpose map}:
\begin{align}
\mathcal{P}_{B\to BC}(X):=\rho_{BC}^{\frac{1}{2}}\rho_B^{-\frac{1}{2}}X \rho_B^{-\frac{1}{2}}\rho_{BC}^{\frac{1}{2}}, \text{ for any } X\in \mathcal{L}(\mathcal{H}_{B}).
\end{align}
Note that in the previous formulas, $\rho_{B}$ and $X$ are understood as elements of $\mathcal{L}(\mathcal{H}_{BC})$ by extending it in the natural way, i.e., $\rho_{B}\otimes \mbox{id}_{C}$. In Ref.~\cite{sut:18} it is shown that the Petz recovery map is indeed a recovery map, i.e., a CPTP map.

Our first result comes as a natural corollary of the previous stated results.

\begin{theorem}
\label{th:1}
Given bipartite marginals $\{\rho_{AB},\rho_{BC}\}$ compatible with a QMC in the order A-B-C, say $\rho_{ABC}$, then the solution of the maximum entropy estimator $\widetilde{\rho}_{ABC}$ is precisely equal to $\rho_{ABC}$. Moreover, $\widetilde{\rho}_{ABC}$ can be algebraically recovered via the Petz map $\mathcal{P}_{B\to BC}(\cdot)$, concretely:
\begin{align}
\widetilde{\rho}_{ABC}=\rho_{BC}^{\frac{1}{2}}\rho_B^{-\frac{1}{2}}\rho_{AB} \rho_B^{-\frac{1}{2}}\rho_{BC}^{\frac{1}{2}}.
\end{align}
\end{theorem}

\begin{remark}We can also recover a tripartite density operator from $\rho_{BC}$ through $\mathcal{P}_{B\to AB}(\cdot)$:
\begin{align}
\rho_{AB}^{\frac{1}{2}}\rho_B^{-\frac{1}{2}}\rho_{BC} \rho_B^{-\frac{1}{2}}\rho_{AB}^{\frac{1}{2}},
\end{align}
and by uniqueness, because the von Neumann entropy is concave, they are the same.	
\end{remark}

In alternative to the Petz recovery map, in presence of marginals compatible with a QMC, we can solve efficiently the associated optimization problem, i.e., determine the Lagrange multipliers in Eq.~\eqref{eq:4} due to the following result by Petz~\cite{pet:03}:
\begin{lemma}
\label{lemma:1}
\cite{pet:03}
Assume that $\rho_{ABC}$ is invertible. Then the equality holds in the strong subadditivity inequality (SSA) if and only if $\log\rho_{ABC}-\log\rho_{AB}=\log\rho_{BC}-\log\rho_B$.
\end{lemma}

Our next result, provides a necessary and sufficient condition for the given marginals $\{\rho_{AB},\rho_{BC}\}$ to be compatible with a QMC. The condition is moreover algebraic, resulting in being easily verifiable.
\begin{theorem}
\label{th:2}
Two bipartite marginals $\{\rho_{AB},\rho_{BC}\}$ are compatible with a QMC in $\mathcal{L}\left(\mathcal{H}_{ABC}\right)$ in the order $A-B-C$ if and only if $\tr_A(\rho_{AB})=\tr_C(\rho_{BC})$ and the operator $\Theta_{ABC}=\rho_{BC}^\frac{1}{2}\rho_{B}^{-\frac{1}{2}}\rho_{AB}^\frac{1}{2}$ is normal.
\end{theorem}
Before giving the proof of Theorem~\ref{th:2}, we will need one further result.
\begin{theorem}
\label{th:3}
\cite{hay:joz:pet:win:04}
A tripartite density operator $\rho_{ABC}\in\mathcal{L}(\mathcal{H}_{ABC})$ satisfies the SSA with equality, i.e., $I_{\rho}(A:C|B)=0$, if and only if there exits a decomposition of the Hilbert space $\mathcal{H}_{B}$ of the form
\begin{align}
\mathcal{H}_{B}=\bigoplus_{j}\mathcal{H}_{B_{j}^{L}}\otimes\mathcal{H}_{B_j^{R}}, \textnormal{ such that } \rho_{ABC}=\bigoplus_{j}p_j\, \rho_{AB_j^{L}}\otimes\rho_{B_j^R C},
\end{align}
with $p_{j}\geq 0$, for all $j$, $\sum_j p_j=1$ and the states $\rho_{AB_{j}^{L}}\in \mathcal{L}\left(\mathcal{H}_{A}\otimes\mathcal{H}_{B_j^{L}}\right)$ and $\rho_{B_{j}^{R}C}\in \mathcal{L}\left(\mathcal{H}_{B_j^R}\otimes\mathcal{H}_{C}\right)$.                 
\end{theorem}
\begin{proof}
$(\Rightarrow)$ Using compatibility, the following decompositions are a direct consequence of Theorem~\ref{th:3} 
\begin{align}
\rho_{AB}&=\bigoplus_{j}p_j\, \rho_{AB_j^{L}}\otimes\rho_{B_j^R}, \textnormal{ with } \rho_{B_j^R}=\tr_{C}(\rho_{B_j^RC}),\\
\rho_{BC}&=\bigoplus_{j}p_j\, \rho_{B_j^{L}}\otimes\rho_{B_j^RC}, \textnormal{ with } \rho_{B_j^{L}}=\tr_{A}(\rho_{AB_j^R}),\\
\rho_{B} &=\bigoplus_{j}p_j\, \rho_{B_j^{L}}\otimes\rho_{B_j^R}.
\end{align}
It is now trivial to see that $\tr_{A}(\rho_{AB})=\tr_{C}(\rho_{BC})$, as required. It remains to check that $\Theta_{ABC}$ is normal. Notice that, from the above decompositions, we have
\begin{align}
\Theta_{ABC} &=\rho_{BC}^{\frac{1}{2}} \rho_{B}^{-\frac{1}{2}}\rho_{AB}^{\frac{1}{2}}\nonumber\\
&=\left(\bigoplus_{j}p_j\, \rho_{B_j^{L}}\otimes\rho_{B_j^RC}\right)^{\!\! \frac{1}{2}}\! \! \! \left(\bigoplus_{k}p_k\, \rho_{B_k^{L}}\otimes\rho_{B_k^R}\right)^{\! \! -\frac{1}{2}}\! \! \! \left(\bigoplus_{l}p_l\, \rho_{AB_l^{L}}\otimes\rho_{B_l^R}\right)^{\!\! \frac{1}{2}}\nonumber\\
&=\bigoplus_{j}p_j\, \rho_{AB_j^{L}}^{\frac{1}{2}}\otimes\rho_{B_j^RC}^{\frac{1}{2}}.
\end{align}
and also
\begin{align}
\Theta_{ABC}^{\dagger}=\left(\rho_{BC}^{\frac{1}{2}}\rho_{B}^{-\frac{1}{2}}\rho_{AB}^{\frac{1}{2}}\right)^{\dagger}=\rho_{AB}^{\frac{1}{2}}\rho_{B}^{-\frac{1}{2}}\rho_{BC}^{\frac{1}{2}}=\bigoplus_{j}p_j\,\rho_{AB_j^L}^{\frac{1}{2}}\otimes\rho_{B_j^RC}^{\frac{1}{2}}=\Theta_{ABC},
\end{align}
hence we conclude that $\Theta_{ABC}$ is self-adjoint, therefore normal.

$(\Leftarrow)$ Consider the operator 
\begin{align}
\varrho_{ABC}=\Theta_{ABC}\, \Theta_{ABC}^{\dagger}=\rho_{BC}^{\frac{1}{2}}\rho_{B}^{-\frac{1}{2}}\rho_{AB}\rho_{B}^{-\frac{1}{2}}\rho_{BC}^{\frac{1}{2}}.	
\end{align}
We have:
\begin{itemize}
	\item[(i)] $\varrho_{ABC}$ is a density operator over $\mathcal{H}_{ABC}$. This can be seen from the fact that the map
	\begin{align}
	\mathcal{P}_{B\to BC}: &  \, \mathcal{L}(\mathcal{H}_{B})\longrightarrow \mathcal{L}(\mathcal{H}_{BC}) \nonumber \\
	&\sigma \mapsto \rho_{BC}^{\frac{1}{2}}\rho_{B}^{-\frac{1}{2}}\sigma \rho_{B}^{-\frac{1}{2}}\rho_{BC}^{\frac{1}{2}}	
	\end{align} 
	is a CPTP map, as mentioned above, and it extends to a CPTP map $\mathcal{I}_{A}\otimes \mathcal{P}_{B\to BC}$ which yields, when applied to $\rho_{AB}$, $\varrho_{ABC}$.
	\item[(ii)] $\varrho_{ABC}$ is compatible with $\{\rho_{AB},\rho_{BC}\}$. Indeed, we have,
	\begin{align}
	\tr_{A}(\varrho_{ABC})=\rho_{BC}^{\frac{1}{2}}\rho_{B}^{-\frac{1}{2}}\tr_{A}\left(\rho_{AB}\right)\rho_{B}^{-\frac{1}{2}}\rho_{BC}^{\frac{1}{2}}=\rho_{BC}^{\frac{1}{2}}\rho_{B}^{-\frac{1}{2}}\rho_{B}\rho_{B}^{-\frac{1}{2}}\rho_{BC}^{\frac{1}{2}}=\rho_{BC},
	\end{align}
	where we have used $\tr_{A}(\rho_{AB})=\rho_B$. Moreover, since $\Theta_{ABC}$ is normal and $\tr_{C}(\rho_{BC})=\rho_B$, we have,
	\begin{align}
	\tr_{C}(\varrho_{ABC})&=\tr_{C}(\Theta_{ABC}\, \Theta_{ABC}^{\dagger})=\tr_{C}(\Theta_{ABC}^{\dagger}\, \Theta_{ABC})=\tr_{C}(\rho_{AB}^{\frac{1}{2}}\rho_{B}^{-\frac{1}{2}}\rho_{BC}\rho_{B}^{-\frac{1}{2}}\rho_{AB}^{\frac{1}{2}})=\rho_{AB}.
	\end{align}
	\item[(iii)] By $(i)$ and $(ii)$ and Definition~\ref{def:qmc}, $\varrho_{ABC}$ is a QMC.
\end{itemize}
\end{proof}
From Theorem~\ref{th:2}, we have the following corollary:
\begin{corollary}
	\label{cor:1}
The space of pairs of bipartite marginals $\{\rho_{AB},\rho_{BC}\}$ which is compatible with a QMC is strictly included in the space of pairs of bipartite marginals $\{\rho_{AB},\rho_{BC}\}$ which are compatible.
\end{corollary}
The last statement seems rather intuitive, but cannot be seen immediately due to the fact that the quantum marginal problem, also for two overlapping marginals, is open. To provide a counter-example, we numerically generated a random $3$-qubit density matrix, from which we obtained two compatible marginals via partial trace. Then, we check the compatibility condition with a quantum Markov chain in Theorem~\ref{th:2}, given the two marginals. We repeat the process until the latter conditions fails.
\section{Quantum Bayesian updating}
\label{section: quantum bayesian updating}
The \emph{principle of minimum discrimination information}~\cite{kul:97,kul:66}, is intrinsically related to the maximum entropy principle and it is at the base of inferential updating of probability distributions. Given a \emph{prior} joint probability distribution $q_{{\bf X}}$ describing the random variables ${\bf X}:=\{X_1...X_n\}$ and an additional set of new information about the system, the most unbiased \emph{posterior} $p_{{\bf X}}$ corresponds to the one that minimizes the Kullback-Leibler divergence with the prior distribution under the constraints given by the additional information. If the prior is the most un-informative one, i.e. the  uniform distribution, then the problem is equivalent to finding the probability distribution that maximizes the Shannon entropy with the given constraints. In Ref.~\cite{van:17}, the choice of the Kullback-Leibler divergence as a functional for inferential updating is explained in light of the maximum entropy principle and some designed criteria.

The learning problem for multipartite quantum states we are proposing here can be stated a multiple-step inferential updating procedure, where starting from the uniform distribution, i.e. the maximally mixed state, the marginals are the additional information set available at each step.

In Ref.~\cite{van:18}, a generalization to the quantum realm of the Bayesian updating procedure is proposed. The Kullback-Leibler divergence is generalized by the von Neumann relative entropy, not by mere replacement, but deriving it from the same designed criteria and the maximum (von Neumann) entropy principle.
\begin{definition}
\label{def:qbu}
(Quantum Bayesian updating) Given a quantum system $\mathcal{X}=\{X_1,...,X_n\}$ described by the joint Hilbert space $\mathcal{H}_{\mathcal{X}}=\bigotimes_{i=1}^{n}\mathcal{H}_{X_i}$, let $\varrho_\mathcal{X}\in\mathcal{L}(\mathcal{H}_{\mathcal{X}})$ be our a prior knowledge about it, and $\{\langle \Theta_j \rangle\}_{j\in J}$ a set of expectation values corresponding to a set of observables $\{\Theta_j\}_{j\in J}$ be the additional set of information. Then, via the principle of minimum updating, the posterior density operator assumes the form
\begin{align}
\label{eq: qbu}
\rho_\mathcal{X}=\exp\left(\lambda \, \mbox{\textnormal{id}}_{\mathcal{X}} + \sum_{j\in J}\alpha_j  \Theta_j +\log \varrho_\mathcal{X}\right),
\end{align} 
where the Lagrange multipliers $\lambda$ and $\{\alpha_j\}_{j\in J}$ are determined by the constraints $\tr(\rho)=1$ and $\tr(\rho\, \Theta_j)=\langle \Theta_j\rangle$, $j\in J$. We indicate the updating process of Eq.~\eqref{eq: qbu} by the diagram
\begin{align}
\varrho_\mathcal{X} \xrightarrow{\, \{\Theta_j,\langle \Theta_j\rangle\}_{j\in J} \,} \rho_\mathcal{X}.
\end{align}
\end{definition}
Observing that $S(\rho_\mathcal{X})=-S(\rho_\mathcal{X}||\unif{\mathcal{X}})$, it follows that the maximum entropy density operator of Eq.~\ref{eq:4} is equivalent to the output of the process:
\begin{align}
\label{eq:27}
\unif{ABC}\xrightarrow{\{\rho_{AB},\rho_{BC}\}}\widetilde{\rho}_{ABC}.
\end{align}
We could split the updating process into two processes, i.e:
\begin{align}
\label{eq:28}
\unif{ABC}\xrightarrow{\rho_{AB}}\widetilde{\sigma}'_{ABC}\xrightarrow{\rho_{BC}}\widetilde{\sigma}_{ABC},
\end{align}
or
\begin{align}
\label{eq:29}
\unif{ABC}\xrightarrow{\rho_{BC}}\widetilde{\varrho}'_{ABC}\xrightarrow{\rho_{AB}}\widetilde{\varrho}_{ABC},
\end{align}
and, classically, one can show, all these updating processes, Eq.~\eqref{eq:27}, Eq.~\eqref{eq:28} and Eq.~\eqref{eq:29}, are equivalent. However, in the quantum scenario, this is not the case due to the non-commutative nature of the observables involved. In Ref.~\cite{sut:faw:ren:16}, it is shown that, in contrast to the classical case, where the conditional mutual information is always a measure of the distance of a tripartite quantum state to a general Markov chain, in the quantum case it is not. In particular,
\begin{align}
I_{\rho}(A:C|B)\geq S(\rho_{ABC}||\mathcal{R}_{B\to BC}(\rho_{AB})),
\end{align}
where $\mathcal{R}_{B\to BC}(\cdot)$ is an arbitrary recovery map. Recently, tighter versions of the above inequality have been the object of study by several authors~\cite{bra:ara:opp:str:15}. The equality holds when $I_{\rho}(A:C|B)= 0$, i.e., for a QMC. 
The fact that the three above processes, Eq.~\eqref{eq:27}, Eq.~\eqref{eq:28} and Eq.~\eqref{eq:29}, lead to different density operators corresponds to the non-commutativity of a diagram. Moreover, in the next theorem we provide necessary and sufficient conditions for the diagram to commute, which provides a characterization of QMCs in terms of QBU.
\begin{theorem}
	\label{th:4}
Given a set of marginals $\{\rho_{AB},\rho_{BC}\}$, we then have a commutative diagram of quantum Bayesian updating processes:
\begin{align*}
\xymatrix@R+1pc@C+2pc{ \unif{ABC} \ar[r]^{\rho_{AB}}\ar[d]_{\rho_{BC}} & \widetilde{\sigma}'_{ABC}\ar[d]^{\rho_{BC}}\\
               \widetilde{\varrho}'_{ABC}\ar[r]_{\rho_{AB}} & \widetilde{\rho}_{ABC}\\
                 }
\end{align*}
if and only if they are compatible with a QMC in the order $A-B-C$.        
\end{theorem}
To provide a proof of this theorem we will need the following lemma.
\begin{lemma}
\label{lemma:2}\cite{kim:12,rus:13}
Given a tripartite quantum state $\rho_{ABC}\in \mathcal{L}(\mathcal{H}_{ABC})$ compatible with $\{\rho_{AB},\rho_{BC}\}$, the following operator-inequalities hold:
\begin{itemize}
	\item [(i)] $\tr_{AB}[\rho_{ABC}(\log\rho_{ABC}-\log\rho_{AB}+\log\rho_B-\log\rho_{BC})]\succeq 0$,
	\item [(ii)] $\tr_{AB}[\rho_{AB}(\log\rho_{AB}-\log\rho_{ABC}-\log\rho_B+\log\rho_{BC})]\succeq 0$,
\end{itemize}
where the equality holds if and only if $\rho_{ABC}$ is a QMC in the order $A-B-C$. Above, $\succeq 0$ stands for positive 	semidefinite.
\end{lemma}

\begin{proof}
$(\Leftarrow)$ We will now show that if $\{\rho_{AB},\rho_{BC}\}$ is compatible with a QMC $\widetilde{\rho}_{ABC}$ in the order $A-B-C$, then $\widetilde{\sigma}_{ABC}=\widetilde{\varrho}_{ABC}=\widetilde{\rho}_{ABC}$.

Using the fact that $S(\rho)=-S(\rho||\unif{ABC})$, the first step of the updating process is obtained by maximizing the von Neumann entropy, yielding
\begin{align}
\widetilde{\sigma}'_{ABC}=\rho_{AB}\otimes\unif{C}	\quad \text{ and } \quad \widetilde{\varrho}'_{ABC}=\unif{A}\otimes \rho_{BC}.
\end{align}
Recalling that $\mbox{Comp}(\rho_{BC}):=\{\rho_{ABC}\in\mathcal{L}\left(\mathcal{H}_{ABC} \right) : \tr_{A}(\rho_{ABC})=\rho_{BC}\}$. Then 
\begin{align}
\widetilde{\sigma}_{ABC}=\underset{\rho_{ABC}\in \mbox{Comp}(\rho_{BC})}
{\mbox{argmin}}S(\rho_{ABC}||\widetilde{\sigma}'_{ABC}).
\end{align}
It then follows that, since $\widetilde{\rho}_{ABC}\in \mbox{Comp}(\rho_{BC})$, 
\begin{align}
\label{eq:33}
S\left(\widetilde{\rho}_{ABC}||\rho_{AB}\otimes \unif{C}\right)\geq S\left(\widetilde{\sigma}_{ABC}||\rho_{AB}\otimes \unif{C}\right).
\end{align}
Analogously, $\mbox{Comp}(\rho_{AB})=\{\rho_{ABC}\in\mathcal{L}\left(\mathcal{H}_{ABC} \right) : \tr_{C}(\rho_{ABC})=\rho_{AB}\}$. Then 
\begin{align}
\widetilde{\varrho}_{ABC}=\underset{\rho_{ABC}\in \mbox{Comp}(\rho_{AB})}
{\mbox{argmin}}S(\rho_{ABC}||\widetilde{\varrho}'_{ABC}).
\end{align}
It then follows that, since $\widetilde{\rho}_{ABC}\in \mbox{Comp}(\rho_{AB})$, 
\begin{align}
\label{eq:35}
S\left(\widetilde{\rho}_{ABC}||\unif{A}\otimes \rho_{BC}
\right)\geq S\left(\widetilde{\varrho}_{ABC}||\unif{A}\otimes \rho_{BC}\right).
\end{align}
The SSA inequality can be obtained as a special case as the contractibility property of the quantum relative entropy under CPTP maps~\cite{lin:73,les:rus:99}:
\begin{align}
S(\rho||\sigma)\geq S(\Phi(\rho)||\Phi(\sigma)), \text{ for all }\rho,\sigma \text{ and } \Phi \text{ a CPT map.}	
\end{align}
Indeed, set $\rho=\rho_{ABC}\in \mbox{Comp}(\rho_{AB})\cap \mbox{Comp}(\rho_{BC})$, then:
\begin{itemize}
	\item [(i)] For $\sigma=\rho_{AB}\otimes \unif{C}$ and $\Phi(\cdot)=\tr_{A}(\cdot)$, we get
	\begin{align}
	\label{eq:37}
		S\left(\rho_{ABC}||\rho_{AB}\otimes \unif{C}\right)\geq S\left(\rho_{BC}||\rho_{B}\otimes \unif{C}\right).
	\end{align}
	\item [(ii)] For $\sigma=\unif{A}\otimes\rho_{BC}$ and $\Phi(\cdot)=\tr_{C}(\cdot)$, we get
	\begin{align}
	\label{eq:38}
		S\left(\rho_{ABC}||\unif{A}\otimes\rho_{BC}\right)\geq S\left(\rho_{AB}||\unif{A}\otimes \rho_{B}\right).
	\end{align}
\end{itemize} 
In both Eq.\eqref{eq:37} and Eq.~\eqref{eq:38} the equality holds for $\rho_{ABC}=\widetilde{\rho}_{ABC}$. Using the inequalities~\eqref{eq:33} and \eqref{eq:37} with $\rho_{ABC}=\widetilde{\sigma}_{ABC}$ and the Eq.\eqref{eq:37} with $\rho_{ABC}=\widetilde{\rho}_{ABC}$,
\begin{align}
& S\left(\widetilde{\rho}_{ABC}||\rho_{AB}\otimes \unif{C}\right)\geq S\left(\widetilde{\sigma}_{ABC}||\rho_{AB}\otimes \unif{C}\right)\geq S\left(\rho_{BC}||\rho_{B}\otimes \unif{C}\right)=S\left(\widetilde{\rho}_{ABC}||\rho_{AB}\otimes \unif{C}\right).
\end{align}
It then follows that $\widetilde{\rho}_{ABC}=\widetilde{\sigma}_{ABC}$.

Analogously, using the inequalities~\eqref{eq:35} and \eqref{eq:38} with $\rho_{ABC}=\widetilde{\varrho}_{ABC}$ and the Eq.\eqref{eq:38} with $\rho_{ABC}=\widetilde{\rho}_{ABC}$, we get $\widetilde{\rho}_{ABC}=\widetilde{\varrho}_{ABC}$.

$(\Rightarrow)$ We will now show that if $\widetilde{\sigma}_{ABC}=\widetilde{\varrho}_{ABC}$ then $\{\rho_{AB},\rho_{BC}\}$ is compatible with a QMC $\widetilde{\rho}_{ABC}$ in the order $A-B-C$ and $\widetilde{\rho}_{ABC}=\widetilde{\sigma}_{ABC}=\widetilde{\varrho}_{ABC}$.

To provide explicit representations of $\widetilde{\sigma}_{ABC}$ and $\widetilde{\varrho}_{ABC}$ as exponentials of some operators we write, using invertibility,
\begin{align}
\rho_{AB} =\exp\left(\sum_{i=0}^{d_{A}^2-1}\sum_{j=0}^{d_{B}^2-1}A_{ij}\Lambda_{i}^{(A)}\otimes \Lambda_{j}^{(B)}\right),\quad \rho_{BC} =\exp\left(\sum_{i=0}^{d_{B}^2-1}\sum_{j=0}^{d_{C}^2-1}C_{ij}\Lambda_{i}^{(B)}\otimes \Lambda_{j}^{(C)}\right),
\end{align}
for $\{A_{ij}\}$ and $\{B_{ij}\}$ real coefficients. For convenience, we also define the partition functions
\begin{align}
Z_{AB}=\exp(-A_{00}) \text{ and } Z_{BC}=\exp(-C_{00}).
\end{align}
  We then have, using Definition~\ref{def:qbu},
\begin{align}
\widetilde{\sigma}_{ABC}=\exp\left(\sum_{i=0}^{d_B^2-1}\sum_{j=0}^{d_C^2-1}\eta_{ij}\Lambda_i^{(B)}\Lambda_j^{(C)}+\log \widetilde{\sigma}'_{ABC}\right).
\end{align}
Where $\{\eta_{ij}\}$ are the Lagrange multipliers in the definition. After some algebra, we can write
\begin{align}
\log\widetilde{\sigma}_{ABC}=&\left(\eta_{00}-\log Z_{BC}-\log d_C\right)\mbox{id}_{ABC}+\sum_{i=0}^{d^2_B-1}\sum_{j=1}^{d_C^2-1}\eta_{ij}\Lambda_i^{(B)}\Lambda_j^{(C)} \nonumber \\
&+\sum_{i=1}^{d_A^2-1}\sum_{j=0}^{d_B^2-1}A_{ij}\Lambda_i^{(A)}\Lambda_j^{(B)}+\sum_{i=1}^{d_B^2-1}(A_{0i}+\eta_{i0})\Lambda_{i}^{(B)}.
\end{align}
Similarly,
\begin{align}
\widetilde{\varrho}_{ABC}=\exp\left(\sum_{i=0}^{d_A^2-1}\sum_{j=0}^{d_B^2-1}\lambda_{ij}\Lambda_i^{(A)}\Lambda_j^{(B)}+\log \widetilde{\varrho}'_{ABC}\right).
\end{align}
Where $\{\lambda_{ij}\}$ are the Lagrange multipliers in the definition. After some algebra, we can write
\begin{align}
\log\widetilde{\varrho}_{ABC}=&\left(\lambda_{00}-\log Z_{AB}-\log d_A\right)\mbox{id}_{ABC}+\sum_{i=1}^{d^2_A-1}\sum_{j=0}^{d_B^2-1}\lambda_{ij}\Lambda_i^{(A)}\Lambda_j^{(B)} \nonumber \\
&+\sum_{i=0}^{d_B^2-1}\sum_{j=1}^{d_C^2-1}C_{ij}\Lambda_i^{(B)}\Lambda_j^{(C)}+\sum_{i=1}^{d_B^2-1}(C_{i0}+\lambda_{0i})\Lambda_{i}^{(B)}.
\end{align}
Since $\widetilde{\sigma}_{ABC}=\widetilde{\varrho}_{ABC}$, by hypothesis, it follows that $\log\widetilde{\sigma}_{ABC}=\log\widetilde{\varrho}_{ABC}$ and, hence,
\begin{align}
& \eta_{00}-\log Z_{BC}-\log d_C=\lambda_{00}-\log Z_{AB}-\log d_A,\\
\label{eq:49}
& A_{ij}=\lambda_{ij},\ \text{for } i=1,...,d_A^2-1 \text{ and } j=0,...,d_B^2-1,\\
\label{eq:50}
& C_{ij}=\eta_{ij},\ \text{for } i=0,...,d_B^2-1 \text{ and }   j=1,...,d_C^2-1,\\
& A_{0i}+\eta_{i0}=C_{i0}+\lambda_{0i},\ \text{for }  i=1,...,d_B^2-1.
\end{align}
Observing that,
\begin{align}
&\sum_{i=1}^{d^2_A-1}\sum_{j=0}^{d_B^2-1}\lambda_{ij}\Lambda_i^{(A)}\Lambda_j^{(B)} =\log\rho_{AB} -\log Z_{AB} \mbox{id}_{AB}-\sum_{i=1}^{d_B^2-1}A_{0i}\Lambda_i^{(B)},\\
&\sum_{i=0}^{d^2_B-1}\sum_{j=1}^{d_C^2-1}\eta_{ij}\Lambda_i^{(B)}\Lambda_j^{(C)}=\log\rho_{BC} -\log Z_{BC} \mbox{id}_{BC}-\sum_{i=1}^{d_B^2-1}C_{i0}\Lambda_i^{(B)},
\end{align}
it follows that the two reconstructed states can be written as
\begin{align}
\log \widetilde{\sigma}_{ABC}=&\log\rho_{AB} +\log\rho_{BC}+\left(\eta_{00}-\log Z_{BC}-\log d_C\right)\mbox{id}_{ABC}+\sum_{i=1}^{d_B^2-1}(\eta_{i0}-C_{i0})\Lambda_{i}^{(B)},\\
\log \widetilde{\varrho}_{ABC}=&\log\rho_{AB} +\log\rho_{BC}+\left(\lambda_{00}-\log Z_{AB}-\log d_A\right)\mbox{id}_{ABC}+\sum_{i=1}^{d_B^2-1}(\lambda_{0i}-A_{0i})\Lambda_{i}^{(B)}.
\end{align}
We now define the Hermitian operator $\theta_B$ by the following two equivalent formulas (consequence of $\log\widetilde{\sigma}_{ABC}=\log\widetilde{\varrho}_{ABC}$ ):
\begin{align}
\log\theta_B &:=\left(\eta_{00}-\log Z_{BC}-\log d_C\right)\mbox{id}_{B}+\sum_{i=1}^{d_B^2-1}(\eta_{i0}-C_{i0})\Lambda_{i}^{(B)}\nonumber\\
&=\left(\lambda_{00}-\log Z_{AB}-\log d_A\right)\mbox{id}_{B}+\sum_{i=1}^{d_B^2-1}(\lambda_{0i}-A_{0i})\Lambda_{i}^{(B)}.
\end{align}
Observe that $\theta_{B}$ can, as usual, be extended by tensoring with the relevant identity maps to the whole Hilbert space $\mathcal{H}_{ABC}$. We conclude that $\widetilde{\sigma}_{ABC}=\widetilde{\varrho}_{ABC}$ implies the existence of an Hermitian $\theta_B$ such that
\begin{align}
\label{eq:56}
\log\widetilde{\sigma}_{ABC}=\log\rho_{AB} +\log\rho_{BC} +\log\theta_B=\log\widetilde{\varrho}_{ABC}.
\end{align}
Plugging in $\rho_{ABC}=\widetilde{\sigma}_{ABC}$ and its logarithm as in Eq.~\eqref{eq:56}, in  Lemma~\ref{lemma:2}, we get:
\begin{itemize}
\item[(i)] $\tr_B[\rho_{BC}(\log\rho_B+\log\theta_B)]\succeq 0$,
\vspace{0.15cm}
\item[(ii)] $-\tr_B[\rho_{B}(\log\rho_B+\log\theta_B)]\mbox{id}_{C}\succeq 0$.
\end{itemize}
The second inequality is equivalent to $\tr_B[\rho_{B}(\log\rho_B+\log\theta_B)]\leq 0$. By tracing over $C$ the first inequality we get $\tr_B[\rho_{B}(\log\rho_B+\log\theta_B)]\geq 0$. Therefore,
\begin{align}
\tr_B[\rho_{B}(\log\rho_B+\log\theta_B)]=0.
\end{align}
From this equation, observing that $\rho_{B}=\tr_{AC}(\widetilde{\sigma}_{ABC})$, we can write
\begin{align}
\tr_{B}[\tr_{AC}(\widetilde{\sigma}_{ABC})(\log\rho_B+\log\theta_B)] =\tr_{ABC}[\widetilde{\sigma}_{ABC} \, \mbox{id}_{A} \otimes (\log\rho_{B}+\log\theta_B)\otimes \mbox{id}_C]=0.
\end{align}
Adding and subtracting both $\log(\rho_{AB})\otimes \mbox{id}_C$ and $\mbox{id}_A\otimes \log(\rho_{BC})$, which for simplicity we write without the identity factors, we get
\begin{align}
\tr_{ABC}[\widetilde{\sigma}_{ABC} \, (\log\widetilde{\sigma}_{ABC}+\log\rho_{B} -\log\rho_{AB} -\log\rho_{BC})]&=0,
\end{align}
and this equation is equivalent to $I_{\rho}(A:C|B)=0$, i.e., equivalent to the statement that the SSA inequality holds with equality.
\end{proof}
\section{Best two out of three}
\label{section: best two out of three}
In Section~\ref{section: marginals' compatibility with qmc}, we provided a way of efficiently reconstructing a tripartite quantum state given two of its bipartite marginals subject to the compatibility condition -- Theorem~\ref{th:1} and Theorem~\ref{th:2}. In general, if we are given the three possible bipartite density operators, we still have a residual degree of freedom, namely, the choice of the pair of bipartite density operators from which one will recover the tripartite estimator. We are going to show that the \emph{best} tripartite maximum von Neumann entropy estimator is the one out of three with \emph{minimum} von Neumann entropy -- Theorem \ref{th:5}.
\begin{theorem}
	\label{th:5}
	Let $\rho\in \mathcal{L}(\mathcal{H}_{ABC})$ be an unknown quantum state describing a tripartite quantum system $ABC$. Given the bipartite marginals $\{\rho_{AB},\rho_{BC},\rho_{AC}\}$, we define $\mathcal{C}=\{\rho_{XY},\rho_{YZ}\}$ with $X,Y,Z\in{A,B,C}$ to be one of the three possible pairs of marginals, and $\widetilde{\rho}_{\mathcal{C}}\in\mbox{Comp}(\mathcal{C})$ to be the maximum von Neumann entropy estimator. Then $\widetilde{\rho}$ minimizing the relative entropy with respect to the unknown state $\rho_{ABC}$ is the one with minimum von Neumann entropy
\begin{align}
\label{eq:60}
\widetilde{\rho}=\mbox{\textnormal{arg}}\,\underset{\mathcal{C}}{\mbox{\textnormal{min}}}\,S\left(\widetilde{\rho}_{\mathcal{C}}\right).
\end{align}
\end{theorem}
\begin{proof}
	By linearity of the trace, the von Neumann relative entropy between $\rho_{ABC}$ and $\widetilde{\rho}_{\mathcal{C}}$:
	\begin{align}
	\label{eq:61}
	S(\rho_{ABC}||\widetilde{\rho}_{\mathcal{C}})=-S(\rho_{ABC}) -\tr[\rho_{ABC}\log\widetilde{\rho}_{\mathcal{C}}].
	\end{align}
	$\widetilde{\rho}_{\mathcal{C}}$  has the form derived in Eq.~\eqref{eq:4}. Then, there exist Hermitian operators $\theta_{XY}\in\mathcal{L}(\mathcal{H}_{XY})$, $\theta_{YZ}\in \mathcal{L}(\mathcal{H}_{YZ})$ and $\theta_{Y}\in\mathcal{L}(\mathcal{H}_Y)$ (naturally extended to act on the joint Hilbert space) such that:
	\begin{align}
	\label{eq:62}
	\log\widetilde{\rho}_{\mathcal{C}}=\theta_{XY}+\theta_{YZ}+\theta_{Y}.
	\end{align} 
	Plugging in ~\eqref{eq:62} in ~\eqref{eq:61} and using the fact that both $\rho_{ABC}$ and $\widetilde{\rho}_{XYZ}$ are in the compatibility set of $\mathcal{C}$, it immediately follows that
		\begin{align}
	\label{eq:63}
	S(\rho_{ABC}||\widetilde{\rho}_{\mathcal{C}})=S\left(\widetilde{\rho}_{\mathcal{C}} \right) -S(\rho_{ABC}).
	\end{align}
	Since the term $S(\rho_{ABC})$ is independent on the choice of $\mathcal{C}$:
	\begin{align}
	\label{eq:64}
	\widetilde{\rho}=\underset{\mathcal{C}}{\mbox{\textnormal{argmin}}}\,S(\rho_{ABC}||\widetilde{\rho}_{\mathcal{C}})=\underset{\mathcal{C}}{\mbox{\textnormal{argmin}}}\,S\left(\widetilde{\rho}_{\mathcal{C}}\right).
	\end{align}
\end{proof}
The estimator we are proposing here is then:
\begin{align}
\label{eq:65}
\widetilde{\rho}=\underset{\mathcal{C}}{\mbox{\textnormal{argmin}}}\,\underset{\rho\in{\small\mbox{Comp}(\mathcal{C})}}{\mbox{\textnormal{max}}}\,S\left(\rho\right).
\end{align}

At this level, the efficiency of the choice of the optimal set of marginals is strictly related to the number of possibilities, which makes the direct generalization to the multipartite scenario inefficient.  As we are going to see in detail in Section~\ref{subsection: the qcl algorithm}, the number of possible choices increases exponentially with the number of variables. The Chow-Liu learning algorithm \cite{cho:liu:68}, solves the problem in the case of probability distributions. The next corollary generalizes the Chou-Liu main argument to QMCs., which will give an hint for the possible generalization of the Chow-Liu algorithm to \emph{Markov quantum trees} (cf. Definition~\ref{def:quantum mf}). Moreover, we are going to see that is sufficient that the compatibility condition with a QMC holds for the estimator obtained via Chow-Liu algorithm to be the optimal one.

\begin{theorem}
	\label{th:6}
	Having a tripartite quantum system $ABC$ described by an unknown $\rho\in \mathcal{L}(\mathcal{H}_{ABC})$ and given the bipartite marginals $\{\rho_{AB},\rho_{BC},\rho_{AC}\}$. If for every pair $\mathcal{C}=\{\rho_{XY},\rho_{YZ}\}$ with $X,Y,Z\in \{A,B,C\}$  there exists a QCM in the order $X-Y-Z$ (cf. Theorem~\ref{th:2}), then the QMC $\widetilde{\rho}_{XYZ}$ minimizing the relative entropy with respect to the unknown state $\rho_{ABC}$ is the one recovered from the pair
	\begin{align}
	\widetilde{\mathcal{C}}=\underset{\mathcal{C}}{\mbox{\textnormal{argmax}}} \Big\{I_\rho(X:Y) + I_\rho(Y:Z)\Big\}.
	\end{align}
\end{theorem}
\begin{proof}
	Observe that:
	\begin{align}
	S(\rho_{ABC}||\widetilde{\rho}_{XYZ})&=\tr\left[\rho_{ABC}\left(\log\rho_{ABC}-\log\widetilde{\rho}_{XYZ}\right)\right]\nonumber\\
	&=-S(\rho_{ABC}) -\tr\left[\rho_{ABC}\left(\log \rho_{XY} +\log \rho_{YZ} -\log \rho_{Y}\right)\right]\nonumber\\
	&=-S(\rho_{ABC}) +S(\rho_{XY})+S(\rho_{YZ})-S(\rho_Y),
	\end{align}
	where in the second line we used Lemma~\ref{lemma:1}. Now, adding and subtracting $S(\rho_{X})$, $S(\rho_{Y})$
	and $S(\rho_Z)$, we immediately get
	\begin{align}
	\label{eq:68}
\!\!\! S(\rho_{ABC}||\widetilde{\rho}_{XYZ})=-\left[I_{\rho}(X:Y)+I_{\rho}(Y:Z)\right]+\!\!\!\!\sum_{W\in \{A,B,C\}}\!\!\!\!S(\rho_{W}) -\!S(\rho_{ABC}).
\end{align} 
	Since the two last terms are independent on the choice of pair of bipartites, we get the desired result.
\end{proof}
Then, in case all the subsets $\mathcal{C}$ are compatible with a QMC, the min-max estimator is the one obtained excluding from the set of marginals the one with lowest quantum mutual information. The compatibility conditions with QMCs are necessary for this result to hold. Indeed, relaxing the compatibility conditions, now we are going to determine an analoguous form to Eq.~\eqref{eq:68} for a general maximum entropy estimator Eq~\eqref{eq:60}. \\

Take Eq~\eqref{eq:62} and add and subtract $S(\rho_{XY})$, $S(\rho_{YZ})$, $S(\rho_Y)$. We then observe that
	\begin{align}
	\textnormal{i)}\ &S(\rho_{XY})+S(\rho_{YZ})-S(\rho_{Y})=-\left[ I_{\rho}(X:Y)+I_{\rho}(Y:Z)\right] +\!\!\!\sum_{W\in \{A,B,C\}}\!\!\! S(\rho_{W}),\\	\textnormal{ii)}\ &\tr\left[\rho_{ABC}\log \widetilde{\rho}_{XYZ}\right]=\tr\left[\rho_{ABC}\left(\theta_{XY}+\theta_{YZ}+\theta_{Y}\right)\right]\nonumber\\
	&=\tr_{XY}\left[\tr_{Z}(\rho_{ABC})\theta_{XY}\right]+\tr_{YZ}\left[\tr_{X}(\rho_{ABC})\theta_{YZ}\right]+\tr_{Y}\left[\tr_{XZ}(\rho_{ABC})\theta_{Y}\right]\nonumber\\
	&=\tr_{XY}\left[\tr_{Z}(\widetilde{\rho}_{XYZ})\theta_{XY}\right]+\tr_{YZ}\left[\tr_{X}(\widetilde{\rho}_{XYZ})\theta_{YZ}\right]+\tr_{Y}\left[\tr_{XZ}(\widetilde{\rho}_{XYZ})\theta_{Y}\right]\nonumber\\
\label{eq:70}
	&=-S(\widetilde{\rho}_{XYZ});\\
\textnormal{iii)}\ &S(\widetilde{\rho}_{XYZ}) -S(\rho_{XY})-S(\rho_{YZ})+S(\rho_Y)= I_{\rho}\left(X:Z|Y \right) 
	\end{align}

Therefore Eq~\eqref{eq:60} can be rewritten as it follows
\begin{align}
\label{eq:72}
S(\rho_{ABC}||\widetilde{\rho}_{XYZ})=&-\left[ I_{\rho}(X:Y)+I_{\rho}(Y:Z)\right]-I_{\rho}\left(X:Z|Y \right) +\sum_{W\in \{A,B,C\}}\!\!\!\! S(\rho_{W})-S(\rho_{ABC}).
\end{align}

Comparing Eq.~\eqref{eq:72} with Eq.~\eqref{eq:68}, we can see that, in general, the choice of the two marginals with maximum mutual informations between the subparts, is not the optimal one for the maximum entropy estimator. Relaxing the compatibility conditions between just one pair with a QMC, the additional term in Eq.\eqref{eq:72}, namely the quantum conditional mutual information of the constructed maximum entropy estimator, $I_{\rho}(X:Z|Y)$, is different from zero. The result of Eq~\eqref{eq:72} does not lead, at first sight, to a simplification of Eq.~\eqref{eq:64}.
\section{The multipartite case}
\label{subsection: the qcl algorithm}
In the previous sections, while considering the tripartite case, we learned that in order to have an algebraic recovery procedure, we need to have a tree structure and an additional constraint regarding the conditional mutual information. It is then natural to suppose that in the general multipartite case one would need a set of additional constraints, which generalize the one obtained previously, and this motivates the following definitions.

\begin{definition}
	\label{def:quantum g}
	(Quantum graph and quantum tree) A quantum graph is a triple $(\mathcal{X},\{\mathcal{H}_{X}\}_{X\in \mathcal{X}}, \rho,\mathcal{G})$, where $\mathcal{X}=\{X_1,..,X_n\}$ labels quantum systems described by the associated Hilbert spaces $\mathcal{H}_{X}$, $X\in \mathcal{X}$, with the $n$-partite composite system described by $\mathcal{H}_{\mathcal{X}}=\mathcal{H}_{X_1}\otimes ...\otimes \mathcal{H}_{X_n}$, $\rho\in\mathcal{L}(\mathcal{H}_{\mathcal{X}})$ is an $n$-partite density operator and $\mathcal{G}=(\mathcal{X}, E)$ is an undirected graph.Whenever the underlying graph of a quantum graph is a tree, we call the structure a quantum tree.
\end{definition}
\begin{definition} 
	Let $\mathcal{G}=(V,E)$ be a undirected graph and let $U$ and $W	$ be non overlapping subsets of $V$. We say that a subset $Z$, disjoint from $U$ and $W$, separates $U$ and $W$ if every path connecting a vertex in $U$ and a vertex in $W$ necessarily overlaps with $Z$. We say that $Z$ is a separator for $U$ and $W$.
\end{definition}
\begin{definition}
	\label{def:quantum mf}
	(Markov quantum field and Markov quantum tree) A Markov quantum field is a quantum graph $(\mathcal{X},\{\mathcal{H}_{X}\}_{X\in \mathcal{X}}, \rho,\mathcal{G})$, where the density operator $\rho$ satisfies the global Markov property: for all $U,W\subset \mathcal{X}$ such that there exists a separator $Z$ for $U$ and $W$ then
	\begin{align}
	\label{CI}
	I_{\rho}(U:W|Z)=0.
	\end{align}  	
Whenever the underlying graph of a Markov quantum field is a tree, we call the structure a Markov quantum tree.
\end{definition}
Definition~\ref{def:quantum mf} includes the most general Markov property, i.e. the global property. Observe that given a Markov quantum field, then any quantum subtree is a quantum Markov tree.

Let $\rho_{\mathcal{X}}\in\mathcal{L}\left(\mathcal{H}_\mathcal{X} \right)$ be an unknown density operator that describes an $n$-partite physical system labelled by $\mathcal{X}$. Let $\mathcal{C}_{\mathcal{T}}=\{\rho_{X_iX_j}, \{X_i,X_j\}\in E(\mathcal{T})\}$ a subset of bipartite marginals, graphically representable by one of its spanning trees $\mathcal{T}$. The quantum state $\widetilde{\rho}_{\mathcal{C}_{\mathcal{T}}}$ that maximizes the von Neumann entropy under the constraints of compatibility with the marginals in $\mathcal{C}_{\mathcal{T}}$ has the form derived in Eq.~\eqref{eq: merho for 2}. Then, there exists a set of Hermitian operators $\{\theta_{X_i}\in\mathcal{L}(\mathcal{H}_{X_i}), X_i\in V(\mathcal{T})\}$ and $\{\theta_{X_iX_j}\in\mathcal{L}(\mathcal{H}_{X_iX_j}), \{X_i,X_j\}\in E(\mathcal{T})\}$, naturally extended to act on the joint Hilbert space, such that:
	\begin{align}
	\label{eq:75}
	\log\widetilde{\rho}_{\mathcal{C}_{\mathcal{T}}}=\sum_{\{X_i,X_j\}\in E(\mathcal{T})}\theta_{X_iX_j}+\sum_{i=1}^{n} \theta_{X_i},
	\end{align}  
where $\deg X_i$ is the \emph{degree of the node} $X_i$, i.e. the number of edges linked to the node. If the quantum tree is a Markov quantum tree, then we have,
	\begin{align}
\label{eq:76}
\log\widetilde{\rho}_{\mathcal{C}_{\mathcal{T}}}=\sum_{\{X_i,X_j\}\in E(\mathcal{T})}\log\rho_{X_iX_j}-\sum_{i=1}^{n}\left( \deg X_i-1\right)\log\rho_{X_i},
\end{align} 

The combinatorial factor is obtained by considering the spanning Markov quantum tree $\mathcal{T}$ over $\mathcal{X}$ as a tripartite one on $ABC$ with $A=X_i\in \mathcal{X}$, $B=X_j\in\mathcal{X}$, $C=\mathcal{X}\backslash\{X_i,X_j\}$, where $X_i$ and $X_j$ are chosen such that there are no edges between $X_i$ and any vertices in $\mathcal{X}\backslash\{X_i,X_j\}$, i.e., $A-B-C$ is a Markov quantum tree, for which we can use Lemma~\ref{lemma:1}. 
Then, consider the Markov quantum subtree $\mathcal{X}\backslash\{X_i,X_j\}$ as a tripartite quantum tree $A'-B'-C'$  with $A=X_k\in \mathcal{X}\backslash X_i$, $B=X_l\in\mathcal{X}\backslash X_i$, $C=\mathcal{X}\backslash\{X_i,X_k,X_l\}$, apply Lemma~\ref{lemma:1}, and iterate the procedure until the remaining subgraph is bipartite. It is then clear that each vertex comes with a factor of its degree minus one.

The construction above allows for an algebraic recovery of the state by iteratively applying the Petz recovery map.
 In the following paragraph, we show that the global Markov condition results in an efficient choice of the best tree.

\paragraph{Choosing the best tree} Given an unknown $n$-partite quantum system, its bipartite correlations can be represented by a complete graph. Its number of possible spanning trees, i.e., tree subgraphs which include all vertices in the graph, is given by Cayley's formula~\cite{cay:89}, $n^{n-2}$, which grows exponentially with the number of vertices. Because of this, we can not choose the best tree efficiently. In the classical scenario, one possible solution is to use the Chow-Liu algorithm~\cite{cho:liu:68}.

Recall that the Chow-Liu algorithm, cf. Appendix, provides an efficient way to find the optimal tree minimizing the Kullback-Leibler divergence between the actual probability distribution, $p(X_1=x_1,...,X_n=x_n)$, and a probability distribution associated to a spanning tree, $p^{\mathcal{T}}(X_1=x_1,...,X_n=x_n)$. In~\cite{cho:liu:68}, it is shown that the Kullback-Leibler divergence can be written as:
\begin{align}
\label{eq:73}
D_{K}(p, p^{\mathcal{T}}) =-\sum_{\{X_i,X_j\}\in E(\mathcal{T})}I(X_i,X_j)	 +\sum_{i=1}^{n}H(X_i)-H(X_1,...,X_n),
\end{align}
where $H(\cdot)$ is the Shannon entropy and $I(X_i,X_j)$ is the classical mutual information between $X_i$ and $X_j$.
The only term that depends on the choice of tree $\mathcal{T}=(V(\mathcal{T})=\mathcal{X},E(\mathcal{T}))$ is the first one, therefore, minimizing the Kullback-Leibler divergence is equivalent to maximizing
\begin{align}
\sum_{\{X_i,X_j\}\in E(\mathcal{T})}I(X_i,X_j),
\end{align}
to which problem the Chow-Liu algorithm provides an efficient solution.

The relative entropy  between the unknown density operator and one of its quantum tree maximum entropy estimators can be written as
\begin{align}
\label{eq:78}
S\left(\rho_{\mathcal{X}}||\widetilde{\rho}_{\mathcal{C}_\mathcal{T}}\right)= -S\left(\rho_{\mathcal{X}}\right) +\tr\left(\rho_{\mathcal{X}}\log\widetilde{\rho}_{\mathcal{C}_{\mathcal{T}}} \right)=S\left(\widetilde{\rho}_{\mathcal{C}_{\mathcal{T}}}\right)-S\left(\rho_{\mathcal{X}}\right).
\end{align}
Where have used the compatibility conditions to perform a calculation similar to that of Eq.~\eqref{eq:70}. It then follows that Eq.~\eqref{eq:65} still holds in the $n$-partite scenario. An alternative form for Eq.~\eqref{eq:78}, which generalizes the one obtained by Chow and Liu for probability distributions Eq.~\eqref{eq:73}, is now derived. Adding and subtracting to Eq.~\eqref{eq:78} the terms $\sum_{\{X_i,X_j\}\in E(\mathcal{T})}S\left( \rho_{X_iX_j}\right)$ and $\sum_{i=1}^{n}\left( \textnormal{deg}X_i-1\right) S\left( \rho_{X_i}\right)$, observing that
\begin{align}
\sum_{\{X_i,X_j\}\in E(\mathcal{T})}S\left( \rho_{X_iX_j}\right)-\sum_{i=1}^{n}\left( \textnormal{deg}X_i-1\right) S\left( \rho_{X_i}\right)=-\sum_{\{X_i,X_j\}\in E(\mathcal{T})}I_\rho(X_i,X_j)+\sum_{i=1}^{n}S\left( \rho_{X_i}\right),
\end{align}
and setting
\begin{align}
\Delta S(\widetilde{\rho}_{\mathcal{C}_{\mathcal{T}}}):=\sum_{\{X_i,X_j\}\in E(\mathcal{T})}S\left( \rho_{X_iX_j}\right)-\sum_{i=1}^{n}\left( \textnormal{deg}X_i-1\right) S\left( \rho_{X_i}\right) -S(\widetilde{\rho}_{\mathcal{C}_{\mathcal{T}}}),
\end{align}
Eq.~\eqref{eq:78} assumes the form
\begin{align}
&S\left(\rho_{\mathcal{X}}||\widetilde{\rho}_{\mathcal{C}_{\mathcal{T}}}\right)=-\sum_{\{X_i,X_j\}\in E(\mathcal{T})}I_\rho(X_i,X_j)-\Delta S(\widetilde{\rho}_{\mathcal{C}_{\mathcal{T}}})+\sum_{i=1}^{n}S\left( \rho_{X_i}\right) -S\left(\rho_{\mathcal{X}} \right).
\end{align}
When $\Delta S (\widetilde{\rho}_{\mathcal{C}})=0$, again the best tree is the one that maximizes the term
\begin{align}
\sum_{\{X_i,X_j\}\in E(\mathcal{T})}I_\rho(X_i,X_j).
\end{align}
This last problem can be efficiently solved using the Chow-Liu algorithm where the mutual information on the bipartite subparts is replaced by its quantum generalization.

Next, we are going to study the conditions the correlations inside the given system have to  satisfy in order to have $\Delta S (\widetilde{\rho}_{\mathcal{C}_{\mathcal{T}}})=0$. Consider the following iterative construction. Since $\mathcal{T}$ is a tree, there exists a leaf, i.e., a vertex with degree $1$ call it $X_{l_{1}}\in\mathcal{X}$ and denote by $\mbox{ad}(X_{l_1})$ its adjacent vertex.	 Set $\mathcal{T}_1=(V_1,E_1):=\mathcal{T}$. Now consider the chain $X_{l_1} - \mbox{ad}(X_{l_1})- V_{1}\backslash\{X_{l_1},\mbox{ad}(X_{l_1})\}$. The associated quantum conditional mutual information reads
\begin{align}
I_{\rho}(X_{l_1}: V_1\backslash\{X_{l_1},\mbox{ad}(X_{l_1})\}|\mbox{ad}(X_{l_1}))= S\left(\rho_{X_{l_1}\mbox{ad}(X_{l_1})}\right)+S\left(\rho_{V_1\backslash\{X_{l_1}\}}\right) -S\left(\rho_{\mbox{ad}(X_{l_1})}\right) -S(\rho_{V_1}).
\label{eq:83}
\end{align}
Observe that $\rho_{V_1}=\widetilde{\rho}_{\mathcal{C}_{\mathcal{T}}}$. Set $\mathcal{T}_{2}:=(V_2,E_2)$, where $V_2=V_1\backslash \{X_{l_1}\}$ and $E_2$ is obtained naturally from $E_1$ by dropping $\{X_{l_1},\mbox{ad}(X_{l_1})\}$. It is trivial to see that $\mathcal{T}_2$ is a tree and, thus, we can find a leaf $X_{l_2}\in V_2$. Consider now the chain $X_{l_2} - \mbox{ad}(X_{l_2})- V_{2}\backslash\{X_{l_2},\mbox{ad}(X_{l_2})\}$. The associated quantum conditional mutual information reads
\begin{align}
I_{\rho}(X_{l_2}: V_2\backslash\{X_{l_2},\mbox{ad}(X_{l_2})\}|\mbox{ad}(X_{l_2}))= S\left(\rho_{X_{l_{2}}\mbox{ad}(X_{l_{2}})}\right)+S\left(\rho_{V_2\backslash\{X_{l_{2}}\}}\right) -S\left(\rho_{\mbox{ad}(X_{l_{2}})}\right) -S(\rho_{V_2}).
\label{eq:84}
\end{align}
It is now clear that $S\left(\rho_{V_1\backslash\{X_{l_1}\}}\right)$ in Eq.~\eqref{eq:83} cancels with $S(\rho_{V_2})$ upon summing the two equations. Iteratively, we can build a tree $\mathcal{T}_{i+1}=(V_{i+1},E_{i+1})$ from a tree $\mathcal{T}_{i}=(V_i,E_i)$  with a chosen leaf $X_{l_{i}}$, by setting $V_{i+1}=V_{i}\backslash\{X_{l_i}\}$ and dropping the edge $\{X_{l_i},\mbox{ad}(X_{l_i})\}$ from $E_i$ to obtain $E_{i+1}$. This construction can be performed until one obtains the last chain $X_{l_{n-2}}$
\begin{align}
&\sum_{i=1}^{n-2}I_{\rho}(X_{l_i}: V_i\backslash\{X_{l_i},\mbox{ad}(X_{l_i})\}|\mbox{ad}(X_{l_i}))\nonumber\\
&= \sum_{i=1}^{n-2} \left[S\left(\rho_{X_{l_i}\mbox{ad}(X_{l_i})}\right)-S\left(\rho_{\mbox{ad}(X_{l_i})}\right)\right] +S\left(\rho_{V_{n-2}\backslash\{X_{l_{n-2}}\}}\right)-S(\rho_{V_1}) \nonumber\\
&= \sum_{\{X_i,X_j\}\in E(\mathcal{T})} S(\rho_{X_iX_j})-\sum_{i=1}^{n}(\deg(X_i)-1)S(\rho_{X_i}) -S(\widetilde{\rho}_{\mathcal{C}_{\mathcal{T}}}),
\end{align}
where we noticed that $V_{n-2}\backslash\{X_{l_{n-2}}\}$ is the last edge missing in the sum on the second line and also that $\bigcup_{i=1}^{n-2}\{\mbox{ad}(X_{i})\}$ is $\mathcal{X}$ except two vertices that have degree one. Therefore, we have
\begin{align}
\Delta S(\widetilde{\rho}_{\mathcal{C}_{\mathcal{T}}}) =\sum_{i=1}^{n-2}I_{\rho}(X_{l_i}: V_i\backslash\{X_{l_i},\mbox{ad}(X_{l_i})\}|\mbox{ad}(X_{l_i})).
\end{align} 
This equation shows that $\Delta S(\widetilde{\rho}_{\mathcal{C}_{\mathcal{T}}})\geq 0$, because each term in the sum is non-negative. To proceed, we will apply the chain rule for quantum conditional mutual information:
\begin{align}
I_{\rho}(A:C_{1}\dots C_{n}| B)=& \sum_{j=1}^{n} I_{\rho}(A:C_{j}|B C_{1}\dots C_{j-1}),
\end{align} 
where the first term in the sum is defined to be $I_{\rho}(A:C_1|B)$. Motivated by the order appearing in the Chain rule, we introduce an order in $V_i\backslash\{ X_{l_i},\mbox{ad}(X_{l_i})\}=\{X_{w_{i}(j)}\}_{j=1}^{n-i}$, so that
\begin{align}
\Delta S(\widetilde{\rho}_{\mathcal{C}_{\mathcal{T}}}) &=\sum_{i=1}^{n-2} I_{\rho}(X_{l_i}: X_{w_{i}(1)}\dots X_{w_i(n-i)}|\mbox{ad}(X_{l_i}))\nonumber \\
&=\sum_{i=1}^{n-2}\sum_{j=1}^{n-i} I_{\rho}(X_{l_i}: X_{w_{i}(j)}|\mbox{ad}(X_{l_i})X_{w_i(1)}\dots X_{w_i(j-1)}).
\end{align}
To have $\Delta S(\widetilde{\rho}_{\mathcal{C}_{\mathcal{T}}}) =0 $ is equivalent to having each term in the sum zero.
\begin{proposition}
\label{prop:3}
The quantity $\Delta S (\widetilde{\rho}_{\mathcal{C}_{\mathcal{T}}})$ is always non-negative. It is zero iff
\begin{align}
I_{\rho}(X_{l_i}: X_{w_{i}(j)}|\mbox{ad}(X_{l_i})X_{w_i(1)}\dots X_{w_i(j-1)})=0, & \text{ for all } j=1,...,n-i \text{ and } i=1,...,n-2.
\end{align}
\hfill\squareforqed
\end{proposition}
We conclude that if we are given a Markov quantum field, then the choice of the best tree is efficient and the recovery procedure is algebraic. Notice that the set of conditions obtained in Prop.~\ref{prop:3} is polynomial, namely $\mbox{O}(n^2)$, and these are in general weaker than the global Markov property. This can perhaps be used as a hint towards relaxing the computationally demanding verification of the global Markov property on the provided marginals. 
                                                                                                              
\section{Conclusions and outlook}
\label{section: conclusions and outlook}
In this manuscript, we proposed a way to compress a subset of density operators according to a generalization to the quantum realm of the Jaynes' max entropy principle, given a chosen set of partial information. Focusing on the tripartite case, with access to two bipartite marginals, we provided a necessary and sufficient algebraic condition for compatibility with a QMC. The recovery procedure through the Petz map is algebraic and efficient. The recovery procedure goes as follows: 
\begin{itemize}
	\item[1.] Measure the three bipartite marginals $\rho_{AB},\rho_{BC},\rho_{AC}$;
	\item[2.] Check for every couple $\rho_{XY}, \rho_{YZ}$ with $X,Y,Z\in \{A,B,C\}$ the compatibility condition with a quantum Markov chain $X-Y-Z$ (Theorem~\ref{th:2});
	\item[3.] If for all the three couples of marginals the compatibility holds, compute the quantum mutual information $I(X:Y)$ and discard the $\rho_{XY}$ with minimum $I(X:Y)$;
	\item[4.] From the two remaining marginals, via Petz recovery map, construct the min-max tripartite estimator (Theorem~\ref{th:1}).
\end{itemize}

In Theorem~\ref{th:6}, we provided a new characterisation of QMCs in terms of a commutative diagram of quantum Bayesian updating processes. This hints on a possible category-theoretical characterization of QMCs, which requires further investigation.

Through the notion of a Markov quantum tree, we were able to generalize the Chow-Liu algorithm to density operators. In fact, the results of this manuscript indicate that the classical theory of learning probability distributions via maximum entropy estimation can be extended naturally to Markov quantum trees.

We speculate that, due to the additional term in Eq.~\eqref{eq:72} (see the paragraph after), it might be possible to extend our results to approximate quantum Markov chains (cf. Ref~\cite{sut:18}), however, one would have to understand if it is possible to have an efficient compatibility condition, i.e., the analogue of Theorem~\ref{th:2} for the case of approximate quantum Markov chains. 

We would like to be able to relax the global Markov condition in Definition~\ref{def:quantum mf} and the result of Proposition~\ref{prop:3} provides a hint that such relaxation might be possible. This is desirable because, in general, the global Markov property seems computationally demanding, for a classical computer, to be verified. Another approach, would be to understand if a quantum computer can learn an even wider class of density operators efficiently.

Another possible direction towards extending the space of learnable quantum states is to enlarge the Hilbert space by an ancilla, which would allow to have a QMC, thus, efficiently learnable, however, this ancilla would have to be subject to certain conditions in order not to end up with a state which would be far from the unknown state in the relative entropy sense.

\section*{Acknowledgements}

This work was supported by European funds namely via H2020 project SPARTA, national funds through FCT, Funda\c{c}\~ao para a Ci\^{e}ncia e a Tecnologia, under contract IT (UID/EEA/50008/2019), PREDICT (PTDC/CCI-CIF/29877/2017), project QuantMining (POCI-01-0145-FEDER-031826) and internal IT projects QBigData and RAPID. Serena Di Giorgio also acknowledges the FCT PhD grant PD/BD/114332/2016 and the DP-PMI FCT programme.

\section*{Appendix: The Chow-Liu algorithm}
\label{appendix: Chow-Liu algorithm}
Given a set of random variables $V=\{X_1,...,X_n\}$ for which we have access to the bipartite correlations described by the joint probability distributions $\{p_{X_iX_j}\}$ we can build a weighted complete graph $\mathcal{G}$ whose vertices label the random variables and the edges correspond to the bipartite probability distributions, $p_{X_iX_j}$, weighted by the classical mutual informations $I(X_i,X_j)$, given by 
\begin{align*}
I(X_i,X_j)=\sum_{x_i,x_j}p_{X_iX_j}(x_i,x_j)\log \frac{p_{X_i}(x_i)p_{X_j}(x_j)}{p_{X_iX_j}(x_i,x_j)}.
\end{align*}
The Chow-Liu algorithm allows us to efficiently construct the maximum weighted spanning tree. Explicitly, sort the values $\{I(X_i,X_j)=I_{\alpha}\}_{\alpha=1}^{N=\frac{1}{2}n(n-1)}$ in descending order, $I_{1}\geq I_{2}\geq .... \geq I_{N}$, then the algorithm proceeds to build a tree $\mathcal{T}$ iteratively as follows:
\begin{itemize}

\item[] \textnormal{[Initialization]} $\ \mathcal{G}_0=(V,E_0)$ where $V=\{X_1,...,X_n\}$ and $E_0=\emptyset$.
\item[] \textnormal{[Iterative Step]} Let $\{X_i,X_j\}$ be the pair associated to $\alpha\in \{1,...,N\}$. Build a graph $\mathcal{G}_{\alpha}=(V,E_{\alpha-1})$, where $E_{\alpha}$ is obtained as follows
\begin{align*}
E_{\alpha}=\begin{cases} 
E_{\alpha-1}\cup \{\{X_i,X_j\}\}, \text{ if }\mathcal{G}_{\alpha}=(V,E_{\alpha}) \text{ is a tree},\\
E_{\alpha-1}, \text{ otherwise.}
\end{cases}
\end{align*}

\end{itemize}
The graph $\mathcal{G}_{N}=\mathcal{T}$ is the desired maximum weighted tree.

\bibliographystyle{unsrt}
\bibliography{bib}
\end{document}